\documentclass[aps,pra,twocolumn,superscriptaddress]{revtex4}

\usepackage{graphicx}
\usepackage{color}
\usepackage{dcolumn}
\usepackage{bm}

\usepackage{amsmath}
\usepackage{amsthm}

\usepackage{textcomp}
\usepackage{amssymb}

\usepackage{xspace}
\usepackage{subfigure}

\newtheorem{theorem}{Theorem}
\newtheorem{definition}{Definition}

\newcommand{\qvbar}{{|}}
\newcommand{\qrangle}{{\rangle}}
\newcommand{\qlangle}{{\langle}}
\newcommand{\ket}[1]{\qvbar{#1}\qrangle}
\newcommand{\bra}[1]{\qlangle{#1}\qvbar}

\newcommand{\ketbra}[2]{\ket{#1}\bra{#2}}

\def\one{{\mathchoice {\rm 1\mskip-4mu l} {\rm 1\mskip-4mu l} {\rm
1\mskip-4.5mu l} {\rm 1\mskip-5mu l}}}
\newcommand{\tr}{\mathrm{tr}}

\newcommand{\cA}{{\cal A}}

\newcommand{\cQ}{{\cal Q}}

\newcommand{\rls}{\mathbb{R}}

\providecommand{\ignore}[1]{}

\begin{document}

\title{ Improving Quantum Clocks via Semidefinite Programming }
\author{Michael Mullan}
\author{Emanuel Knill}
\affiliation{National Institute of Standards and Technology} 
\affiliation{University of Colorado at Boulder}
\date{\today}
\begin{abstract}

The accuracies of modern quantum logic clocks have surpassed those of
standard atomic fountain clocks. These clocks also provide a greater
degree of control, because before and after clock queries, we are able to
apply chosen unitary operations and measurements.  Here, we take
advantage of these choices and present a numerical technique designed
to increase the accuracy of these clocks. We use a greedy approach,
minimizing the phase variance of a noisy classical oscillator with
respect to a perfect frequency standard after an interrogation step;
we do not optimize over successive interrogations or the probe
times.  We consider arbitrary prior frequency knowledge and compare
clocks with varying numbers of ions and queries interlaced with
unitary control. Our technique is based on the semidefinite
programming formulation of quantum query complexity, a method first
developed in the context of deriving algorithmic lower bounds.  The
application of semidefinite programming to an inherently continuous
problem like that considered here requires discretization; we derive
bounds on the error introduced and show that it can be made suitably
small. 
\end{abstract}

\maketitle

\section{Quantum Clocks}

\subsection{The Clock Protocol}
\label{sect:clocks_procedure}

Most atomic clocks are designed to lock a noisy classical oscillator to the
resonance of an atomic standard.  Typically, this is accomplished via
the following clock interrogation protocol:
\begin{enumerate}
\item \label{proc1:preparation}
  \textbf{Preparation}: The atomic system is prepared in some initial state.
\item \label{proc1:interrogation}
  \textbf{Query}: The classical
  oscillator and the atomic system interact.  This modifies the atomic
  state in some way that depends on both the resonant atomic frequency
  $\omega_0$ and the frequency of the classical oscillator $\omega$.
\item \label{proc1:measurement}
  \textbf{Measurement}: The atomic system
  is measured and provides some information about $\omega -
  \omega_0$.
\item \label{proc1:correction}
  \textbf{Correction}: The classical oscillator is adjusted based
  on this information, ideally reducing $|\omega - \omega_0|$.
\end{enumerate}
This protocol must be repeated indefinitely, as the noisy classical
oscillator drifts over time.  Furthermore, the information gained in
step~\ref{proc1:measurement} is always incomplete. Consequently, the
frequency of the classical oscillator is never known exactly and must
be described by a probability distribution.
Figure~\ref{fig:flowchart} illustrates how this distribution changes
as the clock is run.  In this probabilistic perspective, our goal is
to maximize our knowledge of the classical oscillator. For the purpose
of maintaining an accurate clock, we minimize the phase variance of
the classical oscillator with respect to the atomic frequency standard
by optimizing over state preparation (step~\ref{proc1:preparation})
and the post-query measurement (step~\ref{proc1:measurement}).  We
also consider interrogations consisting of multiple queries interlaced
with unitary control.
	 
\begin{figure}[ht]
\includegraphics[width=.5\textwidth]{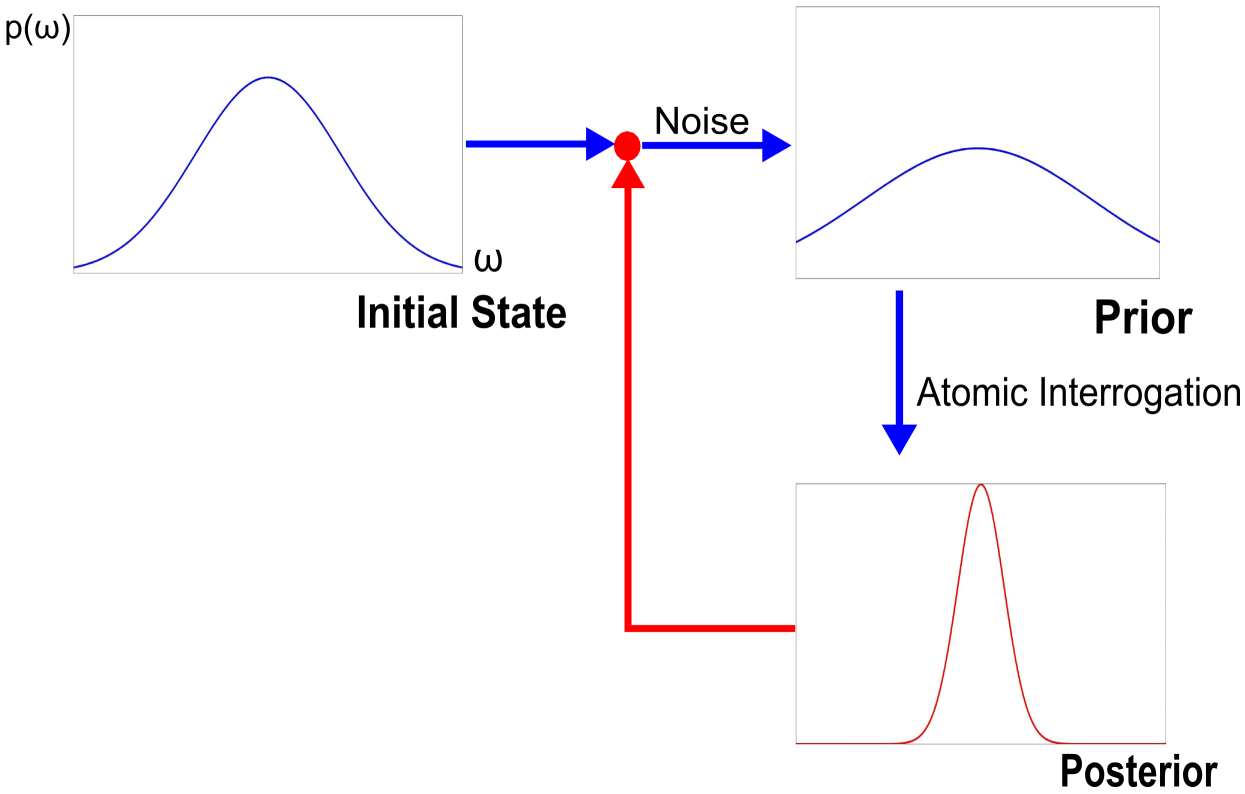}
\caption{Evolution of the classical oscillator's probability
  distribution in a clock protocol.  Noise decreases our knowledge of
  the classical oscillator's frequency, widening $p(\omega)$.  The
  prior probability distribution on the right is used in our
  optimization procedure, and describes the average frequency of the
  classical oscillator over the probe time $T$.  A measurement of the
  atomic standard then yields a measurement outcome $a$, which can be
  used to compute the posterior distribution $p(\omega|a)$. Over many
  iterations of this procedure, we gain knowledge of the history of
  the frequency differences.  This history can be integrated to
  estimate the total time difference between the classical oscillator
  and an ideal atomic clock.  }
\label{fig:flowchart}
\end{figure}

A more complete characterization of a clock involves estimating the
total time difference between the clock and an ideal clock based on
the atomic standard since the clock was started. The time difference
is measured in terms of the total phase difference and requires
integrating the frequency differences over time.  A full Bayesian
treatment of this characterization problem requires that we maintain
the complete history of the classical oscillator by way of the joint
probability distribution, $p( \omega_1, \omega_2 \hdots \omega_t )$,
where the marginals $p(\omega_i)$ reflect our knowledge of the
classical oscillator at a specific time in the past. However, in this
paper we use a greedy approach, namely our optimization procedure does not
take into account this history, and our prior knowledge 
consists only of the marginal $p(\omega_t)$. 
	 
Except for a brief discussion in Sect.~\ref{sect:sdp_noise}, we ignore
the effects of decoherence on the atomic standard and consider only random
fluctuations of the classical oscillator's frequency.  We 
focus on clocks with relatively few atoms and assume full quantum
control of the atoms.  These clocks are often referred to as quantum
clocks. An example is the highly accurate $\text{Al}^+$ quantum clock at NIST
\cite{alClock}, which is a candidate for the application of the
techniques developed here.
 
We consider a simplified model of an $N$-atom clock, whose state is a
superposition of the $N + 1$ symmetric Dicke states
$\ket{0},\ldots,\ket{N}$ of $N$ identical two-level systems. For
example, if $N=2$, the Dicke states are given by
$\ket{0}=\ket{00}_{AB}$,
$\ket{1}=(\ket{01}_{AB}+\ket{10}_{AB})/\sqrt{2}$, and
$\ket{2}=\ket{11}_{AB}$.  Ref.~\cite{buzek} shows that nothing can be
gained by considering other states of the two-level systems.  The $N$
atoms begin in the state $\ketbra{0}{0}$ and are then initialized
(step 1) by the application of a unitary operator $U(0)$ to the state
$\rho(0)=U(0)\ketbra{0}{0}U(0)^\dagger$. A query
(step~\ref{proc1:interrogation}) consists of the application of a
second unitary operator that depends on both $\omega$ and $\omega_0$;
normally, this dependence is only on the difference
$\omega-\omega_0$. Then
\begin{equation} 
\label{eq:oneQuery} 
\rho(0) \rightarrow \rho(0)^{\prime} =  \Omega(\omega-\omega_0) \rho(0)  \Omega(\omega-\omega_0)^{\dagger} .
\end{equation}	 
Finally, the system is measured (step~\ref{proc1:measurement}) with a
positive operator-valued measure (POVM) $\{P_a\}_a$.
	 	 
Interrogation is often done via the Ramsey technique~\cite{ramsey}.
Here, the atoms are subject to two pulses from a classical oscillator
of frequency $\omega$, separated by a period of free evolution of
length $T$.  These pulses are short enough for their dependence on
$\omega$ to be neglected, and so for our purposes, their effect can be
absorbed into the definitions of $U(0)$ and $\{P_a\}_a$.  The period
of free evolution is equivalent to a $z$ rotation by an angle of
$(\omega - \omega_0)T$. Thus, $\Omega(\omega-\omega_0)$ in
Eq.~(\ref{eq:oneQuery}) is given by
\begin{equation}
 \Omega(\omega-\omega_0) = e^{-i J_z (\omega_0- \omega) T}, 
\end{equation}
where $J_z$ is the total $z$ angular momentum operator,
$J_z\ket{k}=(k-N/2)\ket{k}$. By making a global phase change, we can
write the evolution of the Dicke states as
\begin{equation} 
\label{eq:stateEvolution} 
|k\rangle \rightarrow
e^{-ik(\omega-\omega_0)T}|k\rangle. 
\end{equation} 
Without loss of generality, for the remainder of this paper we assume
that $\omega_0 = 0$.

We refer to the action in Eq.~(\ref{eq:oneQuery}) as a ``clock query''.
We can generalize such an action by combining multiple
queries with interlaced unitary operators. The final clock state
can then be written as
\begin{eqnarray} 
\rho(t_f) & = & \Omega(\omega) U(t_f)\hdots \Omega(\omega)U(1) \notag\\
          & & {}\times\rho(0) U^{\dagger}(1) \Omega(\omega)^{\dagger} \hdots U(t_f)^{\dagger} \Omega(\omega)^{\dagger},
\label{eq:clockEvolution}  
\end{eqnarray}
provided that the relative frequencies do not drift between steps 1 and $t_f$. We refer to the complete action in Eq. (\ref{eq:clockEvolution}) as a ``clock interrogation".  As written, $U(t)$ acts only on the clock's atoms; below, we
consider $U(t)$ that can also act on arbitrary ancilla atoms.  Our goal is
to minimize the expected cost
\begin{eqnarray} 
\langle C \rangle 
  &=& \int \int C(\omega - f_a) p(a | \omega ) p(\omega) d\omega da \notag \\
  &=&  \int \int C(\omega - f_a) 
       \tr( P_a \rho_\omega(t_f) ) p(\omega) d\omega da.
\label{eq:clockCost} 
\end{eqnarray}
Here, $f_a$ is the classical frequency estimate for
measurement outcome $a$ with associated POVM operator $P_a$,
$p(\omega)$ is the prior probability distribution of the classical
oscillator's frequency, and $p(a|\omega)$ is the probability of
measurement outcome $a$ given $\omega$ and the interrogation protocol.
The subscript $\omega$ of $\rho$ indicates that the clock was
interrogated by a classical oscillator at frequency $\omega$.  If
$C(\omega - f_a) = (\omega - f_a)^2$, then a minimization of $\langle
C \rangle$ is equivalent to a minimization of the expected posterior
variance of the relative phase change $\omega T$.  To minimize
$\langle C \rangle$ we vary the $U(t)$ in
Eq.~(\ref{eq:clockEvolution}) and the final POVM $\{P_a\}_a$.  The
operator $\Omega(\omega)$ is considered fixed.  

In principle one can consider simultaneously optimizing multiple
sequential interrogations with varying probe times $T$.  Because a
fixed probe time is unable to distinguish between frequencies
differing by multiples of $2\pi/T$, varying the probe time is
necessary to avoid undetected frequency hops. Here, we consider only
one interrogation at a time and fix the probe time $T$. In this case,
there is a scaling symmetry $\omega \rightarrow \alpha\omega$ and
$T\rightarrow \omega/\alpha$, so we fix $T=1$ from now on.
	
\subsection{Background and Summary}

A great deal of research has been done on the theory of atomic clock
optimization.  Of particular interest has been the question of how
quantum effects such as entanglement and squeezing can help overcome
the atomic shot-noise precision limit of $O(1/\sqrt{N})$ and approach
the Heisenberg limit of $O(1/N)$ for $N$ atoms. The possibility that
these effects can result in improved precision was raised
in~\cite{wineland:qc1994a,bollinger:qc1996a}. These and related ideas
are now at the foundations of the subject of quantum
metrology~\cite{giovannetti:qc2011a}. Our work is based on and extends
the analytical studies of Bu\v zek \emph{et al.}~\cite{buzek}, who sought
to optimize clock interrogations under the assumption of a uniform
prior and a family of periodic cost functions.  Starting from results
of Holevo~\cite{holevo:qc1982a}, they obtained a family of initial
states that perform well for large numbers of atoms.  Recently,
Demkovicz-Dobrazanski~\cite{prior} optimized interrogations with costs
determined by the periodic function $C(\omega - f_a) =
4\sin^2\frac{(\omega-f_a)}{2}$ for arbitrary continuous priors. While
this approach is largely analytical, its implementation requires the
numerical maximization of a trace norm. These works focus on
optimizing a single interrogation. Long-term stability has been
analyzed for entangled states of a number of atoms in
Ref.~\cite{winelandBible}, and for a family of squeezing protocols of
atomic ensembles in Ref.~\cite{squeezedClocks}. These studies account
for phase noise in the classical oscillator and optimize clock
protocols given specific feedback mechanisms and noise models.

The periodic cost functions studied in~\cite{buzek,prior} are
convenient for analytic studies of clock optimization but do not
penalize phase errors greater than $2\pi$, even though they correspond
to frequency estimates far from the true frequency of the
oscillator. This issue becomes important when multiple interrogations
and long-term clock stability are considered. Given that we do not
explicitly consider either, and that in the model described above,
phases differing by multiples of $2\pi$ cannot be distinguished, this
may seem irrelevant. Specifically, when optimizing interrogations
consisting of one query or multiple queries with identical probe times
as we do here, the prior and the cost function can in principle be
folded into the interval $[-\pi,\pi]$. However, this folding results
in a cost function that depends on the prior. Because of this and in
view of future extensions of this work, we consider non periodic cost
functions, particularly the quadratic one.

The goal here is to apply the semidefinite programming strategy
originally developed for quantum query algorithms~\cite{BSS,barnumSDP}
to the problem of optimizing clock interrogations. This enables the
greatest flexibility in searching for solutions, as both the prior and
the cost function can be arbitrary.  The first major obstacle is that
the unknown parameter of the queries is continuous rather than
discrete and finite as in typical quantum query algorithms.  We
overcome this by showing how to systematically discretize the
parameter spaces while having good control of the discretization
errors.  The second major obstacle is that the size of the
semidefinite program (SDP) grows at a rate of $\Theta\left(t_f |O|^2 (
|Q| + |F| ) \right)$, where $|O|$ is the number of points in our
discretization, $|Q|$ is the number of atoms we are simulating, and
$|F|$ is the number of POVM elements.  This limits the number of atoms
and the level of discretization for which general clock interrogations
can be optimized, depending on available computational resources.  We
show that for small but useful numbers of atoms, the SDPs can be
implemented and solved given current resources.

The semidefinite programming strategy is formulated in full generality
in Sec.~\ref{sec:sdp}. Our version does not restrict the set of
possible queries or the prior over query parameters. We show how to
explicitly reconstruct the algorithm and measurements from the
solution and discuss how to modify the strategy to account for query
noise. This may be of independent interest for interpolating between
quantum and classical query complexity.  In Secs.~\ref{sec:clocksdp}
and \ref{sec:dError}, we specialize the SDPs to the case of quantum
clocks. Here we show how to discretize the parameter spaces to obtain
finite SDPs while bounding discretization errors. In
Sec.~\ref{sec:results} we show the results from applying the
discretized SDPs to concrete clock problems. In particular we compare
the computational results obtained to prior work, demonstrating both
the ability to obtain improvements and to determine bounds on optimal
costs.\\

\section{The Semidefinite Programming Formulation of Quantum Query Complexity}
\label{sec:sdp}

\subsection{Constructing the Semidefinite Program}
\label{subsec:sdp}

The operation of a quantum clock can be expressed naturally in the
query model of quantum computation.  Here, we are given an oracle (or
black box) chosen from some finite set $\{\Omega(1),
\Omega(2)... \Omega(M)\}$. Each oracle $\Omega(x)$ is a unitary
operator, selected with probability $p(x)$.  The goal is to determine
$x$ with queries, that is, with applications of $\Omega(x)$ to quantum
states and measurement. Often in this context, one is interested in
minimizing the number of queries needed to learn $x$ with near
certainty.  Here, however, we fix the number of queries and seek to
minimize the expected difference between an estimate $x^{\prime}$ of
$x$ and $x$ as quantified by a cost function $C(x,x^{\prime})$.

We adapt the semidefinite programming formulation of quantum query
complexity developed by Barnum \emph{et ~al.} \cite{BSS,barnumSDP} in the
context of proving quantum lower bounds. They cast the problem of
determining the number of queries required in terms of a test of
feasibility of a semidefinite program. Here we aim to minimize
an expected cost, which requires optimizing an objective function
with an extended SDP. For a general introduction to semidefinite programming see Ref. \cite{SDP}.

To formulate the SDP, we introduce quantum systems $O$ and $Q$ for the
oracle and the querier, respectively.  System $Q$ is the one on which
the oracle operators $\Omega(x)$ act. An additional system of ancillas
$A$ may be used by the querier; for the purposes of the SDP, we
normally trace out $A$.  We use the convention that the systems on
which an operator acts are denoted by superscripts.  For states
(density operators) $\rho$, partially omitted system superscripts
imply the partial trace over the omitted systems. Superscripts may be
completely omitted if the set of systems being acted on is clear or
irrelevant.  The dimension of the state space of system $S$ is denoted
by $|S|$.  Initially, the state of the oracle system is given by
$\rho_0 = \sum_x \sum_y \sqrt{p(x)} \sqrt{p(y)} \ketbra{x}{y}$. By
representing the oracle probabilities with a pure superposition, we
encode the fact that there is no information about $x$ available to
the querier (or any system other than $O$).  We define the joint
operator $\Omega^{OQ} = \sum_x \ketbra{x}{x} \Omega^{Q}(x)$, which
applies $\Omega(x)$ to system $Q$ conditional on the state $\ket{x}$
of the oracle.  In this setting, a multi-query quantum computation is
given by the composition
\begin{widetext}
\begin{equation} \label{eq:computation} \rho^{OQA}(t_f) =  \Omega^{OQ} U^{QA}(t_f)  \hdots \Omega^{OQ} U^{QA}(1)  \rho^{OQA}(0) U^{\dagger QA} (1)  \Omega^{\dagger OQ} \hdots U^{\dagger QA}(t_f) \Omega^{\dagger OQ} \end{equation}  
\end{widetext}
followed by a measurement of $QA$ with a POVM. Here $\rho^{OQA}(0)$
is the initial state, with $\rho^O(0) = \rho_0$ pure as defined above, and
$\rho^{QA}(0)$ is a state chosen by the querier. The unitary operators
$U^{QA}(t)$ are inter-query operators that can be chosen arbitrarily
for each step $t$. The querier's protocol (or algorithm) is determined by the
initial state $\rho^{QA}(0)$, the $U^{QA}(t)$, and the POVM.
Refs. \cite{BSS,barnumSDP} show that there is a correspondence between
these algorithms and solutions to a set of semidefinite
constraints. We split the constraints into those that correspond to
the initial state and choice of inter-query operators, and those that
correspond to the POVM. The first set, $S_E$, consists of 
\begin{enumerate}
\item $\rho^O(0) = \rho_0 $, 
\item[] and the following for $t \in\{0,1,\ldots, t_f\}$:
\item  $\rho^{OQ}(t) \geq 0$,
\item  $\rho^O(t) = tr_Q( \rho^{OQ}(t) )  $ and
\item  $\rho^O(t) = tr_Q ( \Omega^{\dagger OQ} \rho^{OQ}(t-1) \Omega^{OQ} ) $.
\end{enumerate}

The constraints corresponding to the POVM and the objective to be
optimized can be based on a connection to the concept of remote state
preparation as described next.
  
\subsection{Measurement and Remote State Preparation}
	
Remote state preparation~\cite{GHJW} involves two systems $O$ and $Q$.
Given a joint state of the two systems, we can prepare states of $O$
conditional on measurements of $Q$.

\begin{definition} \label{def:RSP}
We say that we can \emph{remotely prepare} $\{\sigma_a^O\}_a$ from
the state $\rho^{OQ}$ on $O$ and $Q$ if there exists a POVM of $Q$
with operators $\{P_a^Q\}_a$ such that $tr_Q(P_a^Q \rho^{QO}) =
\sigma_a^O$.
\end{definition} 

The members $\sigma_a^O$ of the set that can be remotely prepared are
positive operators with trace $0\leq \tr(\sigma_a^O)\leq 1$.  The
definition implies that the state $\sigma_a^O/\tr(\sigma_a^O)$ can be
conditionally prepared with probability $\tr(\sigma_a^O)$ by means of
a fixed POVM on $Q$. For definiteness, consider a POVM
$\{P_a^Q\}$. Any implementation of the POVM has the desired
effect. Such an implementation is described by a quantum operation with
Kraus operators $E_a$ such that $E_a^\dagger E_a = P_a$.  For outcome
$a$, the unnormalized $OQ$ state is $E_a^Q\rho^{OQ}E_a^{\dagger
  Q}$. Cyclicity of partial trace implies that $\tr_Q
E_a^Q\rho^{OQ}E_a^{\dagger Q} = \tr_Q P_a^Q\rho^{OQ} =
\sigma_a^O$. The probability of the outcome is $\tr(\sigma_a^O)$.

Theorem \ref{the:RSP} characterizes the set of states that can be
remotely prepared according to Def.~\ref{def:RSP} from a pure state.

\begin{theorem} \label{the:RSP} 
Let $\rho^{OQ}$ be a pure state of systems $O$ and $Q$. We can
remotely prepare $\{\sigma_a^O\}_a$ if and only if $\sum_a \sigma^O_a
= \rho^O$.
\end{theorem}

The proof of Thm.~\ref{the:RSP} is given in App.~\ref{app:a}.

In order to apply Thm.~\ref{the:RSP} we formulate the expected cost in
terms of the POVM-conditional unnormalized states $\sigma_a^O$.  In
the clock problem, the querier associates an estimated frequency $f_a$
with each measurement outcome $a$.  Assume for simplicity that the
possible classical oscillator frequencies $\omega$ come from a finite
set.  We can define operators $A_a^O$ in the oracle basis
$\{\ket{\omega}\}_\omega$ by $(A_a)_{\omega,\omega'} =
\delta_{\omega,\omega'} C(\omega-f_a)$. Eq.~(\ref{eq:clockCost}) can
be rewritten as
\begin{equation}
\langle C\rangle = \sum_a \tr(A_a\sigma_a).
\label{eq:genCost}
\end{equation}
With this equation as motivation, we consider the general situation
where the expected cost of a querier protocol is computed from
predetermined cost operators $A_a$ according to the POVM-conditional
oracle states as in Eq.~(\ref{eq:genCost}).  This motivates the
following SDP $S_M(\rho^O)$ for a given oracle state $\rho^O$:
\begin{equation}
\text{Minimize $\tr(\sum_a \sigma_a A_a)$ subject to:} 
 \begin{cases}
 \forall a \sigma_a \geq 0, \\
 \sum_a \sigma_a = \rho^O .
 \end{cases}
\end{equation}

\begin{theorem} \label{the:lowerBound}
Suppose that $\rho^{OQA}$ is pure. Then $S_M(\rho^O)$ computes the minimum
expected cost over measurements on $QA$.
\end{theorem}

\begin{proof} 
It suffices to observe that according to Thm~\ref{the:RSP},
sets $\{\sigma_a\}_a$ satisfying the constraints of $S_M$ are precisely
the sets that can be remotely prepared with access to systems $QA$.
\end{proof}

The complete SDP for the query optimization problem is obtained by
combining $S_Q=S_E\cup S_M(\rho^O(t_f))$. Because query algorithms
have access to the non-$O$ systems of a purification of $\rho^O(t_f)$,
$S_Q$ computes the optimal average cost of a $t_f$-query quantum
algorithm.  The SDP considered in the spectral adversary
method~\cite{BSS,barnumSDP} is a relaxation of $S_Q$ with a modified
objective designed to determine a lower bound on the number of queries
needed to obtain some fixed probability of error.  There has been significant
recent progress in refining this relaxation~\cite{negAdv} and
demonstrating that it is nearly exact~\cite{rAdv}.  Note that in our
case, we can construct cost operators $A_a$ such that $S_Q$ minimizes
the probability of error for fixed $t_f$. 

\subsection{Algorithm and POVM reconstruction}

For any solution of $S_Q$, in particular the optimal one, there is an
explicit query algorithm that achieves the associated cost objective.
In particular, given the sequence of density matrices $\rho^{OQ}(t)$
and the conditional operators $\sigma_a$, it must be possible to infer
the sequence of unitaries in Eq.~(\ref{eq:computation}) and the POVM
achieving the expected cost. The POVM yielding the conditional
operators $\sigma_a$ can be constructed as in the proof of
Thm.~\ref{the:RSP}.  To determine the unitaries $U^{QA}(t)$, first
extend the querier system with an ancilla $A$, where $A$ has dimension
$|O||Q|$. For each $t$, construct pure states $\rho'^{OQA}(t)$ and
$\rho^{OQA}(t+1)$ by purifying $\rho'^{OQ}(t)=\Omega^{\dagger OQ}
\rho^{OQ}(t) \Omega^{OQ}$ and $\rho^{OQ}(t+1)$, respectively. By
construction, $\rho'^{O}(t)=\rho^{O}(t+1)$. One can therefore use the
Schmidt forms for the pure states $\rho'^{OQA}(t)$ and
$\rho^{OQA}(t+1)$ to construct unitaries $U^{QA}(t)$ satisfying
\begin{equation}
 \rho^{OQA}(t+1) = U^{QA}(t) \rho^{'OQA}(t) U^{\dagger QA}(t). 
\end{equation}

\subsection{Noise}
\label{sect:sdp_noise}

The SDP $S_Q$ is based on the assumption that there is no noise in the
query process. In particular, this excludes decoherence during clock
queries. It is possible to adapt $S_Q$ to include the effects of
noise. This requires that we extend the states $\rho^{OQ}$ by
querier-inaccessible systems $E_i$ modeling the environments causing
the noise. The net effect of query and noise can be modeled
by an oracle-conditional isometry defined by 
\begin{equation}
\ketbra{x}{y}\rho^{Q}\rightarrow\ketbra{x}{y}
  \Omega(x)^{QE}\rho^{Q}\ketbra{\epsilon}{\epsilon}\Omega(y)^{\dagger QE},
\end{equation}
where $\ket{x},\ket{y}$ are the standard oracle basis states,
$\ket{\epsilon}$ is a fixed initial state of $E$ and $\Omega(x)^{QE}$
is unitary. Define $\Omega^{OQE}=\sum_x\ketbra{x}{x}\Omega(x)^{QE}$.  In
the absence of noise, $\Omega(x)^{QE}=\bra{x}\Omega^{OQ}\ket{x}$.
In many cases one can decompose
$\Omega(x)^{QE}=D^{QE}\bra{x}\Omega^{OQ} \otimes I^E \ket{x}$ for an
$O$-independent unitary $D^{QE}$. An example is the clock query
in the presence of phase decoherence.  In a sequence of queries, a new
version of $E$, $E_i$ is introduced at each step by the isometry.
Let $E^t=E_1E_2\ldots E_t$. To account for the noisy
query, the SDP $S_Q$ is modified to $S_D$ as follows:
\begin{align*} \text{Min}&\text{imize\ } \tr(\sum_a \sigma_a^{OE^{t_f}} A_a^{O}) \text{\ subject to} \\
 & \sigma_a^{OE^{t_f}} \geq 0 \\ 
 & \sum_b \sigma_b^{OE^{t_f}} =
  \rho^{OE^{t_f}}(t_f) \\ 
 & \rho^{OE^{t}Q}(t) \geq 0 \\ 
 & \rho^{OE{^t}}(t) = tr_Q(\rho^{OE^{t}Q}(t) ) \\ 
 & \rho^{OE^{t}}(t) = tr_Q \left( \Omega^{OQE_t}
  \rho^{OE^{t-1}Q}(t-1) \ketbra{\epsilon}{\epsilon}
  \Omega^{\dagger OQE_t} \right) \\ 
 & \rho^{OE^{0}}(0) =  \rho_0.
\end{align*} 
(There is an implicit ``for all'' over the free variables $a$ and $t$.)

Phase decoherence is a particularly interesting example of this more
general SDP.  As mentioned above, for complete phase decoherence, the
query isometry factors, in this case giving
\begin{equation}
D^{QE}\ket{i}^Q\ket{\epsilon}^E = 
\ket{i}^Q\ket{\epsilon_i}^E,
\end{equation}
where the $\ket{\epsilon_i}$ are orthonormal states. The effect is to
perfectly correlate the environment at each step with the standard
query basis. In the context of standard query algorithms with a cost
function that captures the probability of successfully identifying the
oracle, the optimal solution to the SDP with complete phase
decoherence corresponds to an optimal classical query algorithm for
the given number of queries $t_f$. By modifying $D^{QE}$ to model
incomplete phase decoherence, it is possible to interpolate between
classical and quantum query algorithms, albeit at the large cost of
adding the systems $E_i$. 

Note that it is not possible to simply trace out the $E_i$ in the SDP:
As can be seen from the method of reconstructing the algorithm from a
solution of the SDP, this would be equivalent to giving the querier
access to the $E_i$. That is, the querier has implicit access to
anything that gets traced out, since traced out systems are not
constrained by the SDP. When the noisy query factors, this is
equivalent to not having had any noise at all, because the querier can
just undo the noise isometries.

\section{The Clock SDP}
\label{sec:clocksdp}

One can apply the SDP $S_Q$ to the clock problem, but the result is
not finitely implementable because the oracle system $O$ is continuous
and the cost operators are continuously indexed.  To make the SDP
finite, we discretize both the possible oracle frequencies and the
frequency estimates that determine the cost operators.  In particular,
we constrain $\omega\in\{\omega_1,\ldots,\omega_{d}\}$, and restrict
the measurement outcomes to a finite set $a\in\{1,\ldots,m\}$
associated with the set of frequency estimates $F=\{f_a\}_{a=1,\ldots,m}$. The
frequency estimates need not be among the $\omega_i$.  The query
system $Q$'s Hilbert space is spanned by the Dicke states
$\ket{k}$. The oracle initial state is given by
\begin{equation}
\rho^O(0)=\sum_{x,y}\sqrt{p(\omega_x)p(\omega_y)}\ketbra{x}{y},
\end{equation}
where $p(\omega_x)$ is the prior probability of $\omega_x$, a
discretization of the continuous prior, and $\ket{x}$ is the oracle
basis state corresponding to classical oscillator frequency
$\omega_x$.  The operator $\Omega^{OQ}$ is now given by
\begin{equation}
\Omega^{OQ} = \sum_{x}\ketbra{x}{x}\exp(-iJ_z^Q\omega_x),
\end{equation}
up to irrelevant $x$-dependent phases. As noted before Eq.~\eqref{eq:genCost},
the cost operators are given by
\begin{equation}
(A_a)_{x,y} = C(\omega_x - f_a) \delta(x,y).
\end{equation}
Since the SDP depends on $p$ and $F$, we denote it by
$S_C(p,F)$. We also use $S_C(p,F)$ to denote the optimum achievable
cost given $p$ and $F$. For discretized prior and frequency estimates,
this is the cost computed by the SDP. 

As written, the size of the SDP $S_Q$ is $\Theta\left(t_f |O|^2 (
|Q|^2 + |F| ) \right)$, where $F=\{f_a\}$. Our implementation of the
SDP takes advantage of the fact that for the clock problem, $\Omega$
is diagonal in the Dicke basis, thus reducing the size and therefore the
memory and time resources required to solve the SDP. In particular,
the set of solutions is invariant under the transformation
$\rho^{OQ}(t)\rightarrow U(t)^Q\rho^{OQ}(t)U(t)^{\dagger Q}$ for
operators $U(t)$ diagonal in the Dicke basis. It follows that we can
restrict the SDP by assuming $\rho^Q(t)$ is diagonal.  The restricted
SDP's size is $\Theta\left(t_f |O|^2 ( |Q| + |F| ) \right)$.

\section{Discretization Error}
\label{sec:dError}

When solving the SDP, it is necessary to make an estimate of the
difference between the discretized SDP's optimal cost and that of the
infimum of the costs of solutions to the continuous problem.  We
separately bound the error due to discretizing the prior (\emph{oracle
  discretization}) and that due to discretizing the measurement
outcomes (\emph{querier discretization}).  To deal with the oracle
discretization error, we solve clock SDPs with random oracle
discretizations so that the average cost for different discretizations
is a lower bound on the optimum cost of the continuous problem. An
upper bound is obtained by a cost integral applied to any of the SDPs
solved. For the quadratic cost function, the querier discretization
error can be bounded in a way that depends only on the set
$F=\{f_a\}_a$ and can be made to go to zero by increasing the size and
resolution of $F$. We start with querier discretization.

\subsection{ Querier Discretization Error }
\label{sect:dErrorB}

To guarantee that we determine the optimum cost for a frequency prior
$p$, the set of allowed querier frequency estimates $F$ should be all of
$\rls$. To solve the clock SDP numerically, we restrict $F$ to
a finite subset. By definition, $S_C(p,F)\geq S_C(p,\rls)$.  The goal
of this section is to obtain a lower bound on $S_C(p,\rls)$ depending
on $F$, and to explain an iterative strategy for locally optimizing
$F$.  Let $F=\{f_1,\ldots,f_N\}$ with $f_j<f_{j+1}$.

\begin{theorem}
\label{thm:qBound}
Let $C$ be a second differentiable, non-negative function satisfying
$C''(x)\leq b$, $C(0)=0$ and $C$ is monotone on $[0,\infty)$
and $(-\infty,0]$. Define $M(\omega)$ by
\begin{equation}
M(\omega) = \begin{cases} C(\omega-f_1 ) & \text{for\ } \omega
  \leq f_1 \\ C(\omega - f_N) & \text{for\ } \omega \geq
  f_N \\ 0 & \text{otherwise} \end{cases}
\end{equation}
We have the following inequality:
\begin{eqnarray}
S_C(p,F)-S_C(p,\rls) &\leq&
  \max_j \frac{b}{8}(f_{j+1}-f_{j})^2 \notag\\
  && {}+ 
  \int M(\omega) p(\omega)d\omega.
  \label{eq:qBound}
\end{eqnarray}
\end{theorem}

The proof is in App.~\ref{app:b}. A strategy for obtaining an
initial choice of an $N$-point $F$ is to optimize the right-hand-side
of Eq.~(\ref{eq:qBound}). 

An important question is how many frequency estimates $f_a$ are needed
so that $S_C(p,F)=S_C(p,\rls)$. One can obtain an upper bound by
observing that the rank $r$ of the final oracle density matrix
$\rho^O(t_f)$ is bounded by $|Q|^{t_f}$ and $|O|$, where $|Q|$ and
$|O|$ are the dimensions of the query and oracle systems. From this
one finds that the optimal POVM can always be reduced to at most $r^2$
elements, implying that a set $F$ with $|F|\leq r^2$ suffices.  There
is evidence that a smaller set is optimal for the clock problem. For
$t_f=1$ and for some costs and priors, we need only $|F|=r$ elements.
\cite{measurement,prior}.

At present, we have no provably correct way of choosing the finite set
of frequency estimates optimally. However, for the quadratic cost
function $C(x)=x^2$, the following heuristic is often effective at
optimizing the choices given an oracle discretization.  We begin with
any set of estimates.  We then run the SDP and use the solution to
compute the posterior distribution for each measurement outcome:
\begin{equation} 
  p(\omega|f_a) = 
    \frac{p(f_a, \omega)}{p(f_a)}  
    = \frac{(\sigma_a)_{\omega,\omega}}{ \tr( \sigma_a ) },
    \label{eq:posteriors} 
\end{equation}
where the $\sigma_a$ are the measurement-conditional unnormalized
oracle states at the end of the algorithm.  Next, we compute the mean
of each of these distributions:
\begin{equation} 
\label{eq:posteriorMean} 
\langle \omega|a \rangle
  = \sum_x \frac{ (\sigma_r)_{x,x} }{\tr( \sigma_a) } \,\omega_x,
\end{equation} 
where $x$ indexes the finite set of frequencies and the expression in
angle brackets is the average of $\omega$ given that measurement
outcome $a$ was obtained.  Replacing the original estimates $f_a$ by
their posterior means is guaranteed to improve the cost without having
to change the algorithm.  We then run the SDP again, replacing the
frequency estimate $f_a$ with $\langle \omega|a \rangle$.  This
procedure is repeated until each estimate is numerically close to its
posterior mean.  The procedure can be adapted to other costs, but the
mean must be replaced by a statistic appropriate for the cost. For
example, for $C(x)=|x|$ we compute the median instead of the mean.

\subsection{ Oracle Discretization Error }
\label{sect:dErrorA}

Let $p(\omega)$ be the probability density of the prior distribution
of clock frequencies. For simplicity, we assume that the prior
distribution is absolutely continuous with respect to Lebesgue
measure. Let $P(\omega)$ be the cumulative distribution function of
$p(\omega)$, and $P^{-1}$ its inverse on $(0,1)$.  Given an offset
$o\in(0,1/d)$, we can define the probability distribution
\begin{equation}
p_o(\omega)
= \sum_{k=0}^{d-1} \frac{1}{d}\delta_{P^{-1}(o+k/d)}(\omega).
\end{equation}
This is an instance of a discretized prior.  Define $\omega(k,o)=
P^{-1}(o+k/d)$. The distribution $p_o$ is designed to approximate $p$
as $d$ goes to infinity. If we choose $o$ uniformly at random from
$(0,1/d)$, it is identical to $p$. Specifically,
\begin{align} \label{eq:inverseTransform}
& \int_{o=0}^{1/d} p_o(\omega) p(o) do \notag \\
& = \int_{o=0}^{1/d}d\sum_{k=0}^{d-1}\frac{1}{d}\delta_{\omega(k,o)}(\omega) do =
p(\omega).
\end{align}
This follows from inverse transform sampling \cite{inverseSampling} (pg. 28).

To estimate the optimum cost $S_C(p,F)$, we estimate the average
$\langle S_C(p_o,F)\rangle_{o\in(0,1/d)}$ by solving $S_C(p_o,F)$ for
a number of offsets $o$ chosen uniformly at random from $(0,1/d)$.
According to the next theorem, this gives a lower bound on $S_C(p,F)$.
We assume that the frequency estimates have already been
discretized, $F=\{f_a\}_a$.

\begin{theorem} \label{theorem:lowerBound} The following inequality holds:
$\langle S_C(p_o,F)\rangle_{o}\leq S_C(p,F)$.
\end{theorem}

\begin{proof} 
Consider an arbitrary query algorithm $\cQ$. Given a prior $r(\omega)$,
$\cQ$ results in the measurement-conditional, unnormalized oracle states
$\sigma_a(r,\cQ)$.  Because the cost operators are diagonal, we can
consider just the diagonals of the $\sigma_a(r,\cQ)$, which define the joint
probability distributions $r(a,\omega|\cQ)$. 
Because the oracle operators are conditional on the standard oracle
basis, $r(a,\omega|\cQ)$ factors as
\begin{equation}
r(a,\omega|\cQ) = q(a|\omega, \cQ) r(\omega),
\end{equation}
where, as indicated, the distribution $q$ does not depend on $r$. 

The average cost for $\cQ$ and $r$ is given by
\begin{equation}
C(r,\cQ)  = \sum_a\int C(\omega-f_a)r(a,\omega|\cQ) d\omega.
\end{equation}
If $\cQ_{\min}$ is an optimal algorithm for
$S_C(r,F)$, then for any algorithm $\cQ$, it follows that $C(r,\cQ)\geq
C(r,\cQ_{\min})=S_C(r,F)$.  Let $\cQ_{o}$ and $\cQ_{\text{opt}}$ be optimal
algorithms for $S_C(p_o,F)$ and $S_C(p,F)$, respectively.
Then
\begin{align}
\langle S_C(p_o,F)\rangle_o 
  &=\langle C(p_o, \cQ_o)\rangle_o\notag\\
  &\leq \langle C(p_o, \cQ_{\text{opt}})\rangle_o,
\label{eq:scpoF1}
\end{align}
where the subscript on the expectations indicates that they are taken
with respect to the distribution over offsets $o$.  The intuition here
is that we can obtain a lower cost if we are able to choose a
different algorithm $\cQ_o$ for different choices of $o$, than if we
are forced to use the same algorithm $\cQ_{\text{opt}}$ in all cases.
We can continue from the last line of Eq.~\eqref{eq:scpoF1} as follows:
\begin{align}
\langle C(p_o, \cQ_{\text{opt}})\rangle_o\hspace{-.25in}&\notag\\
  &= \left\langle\sum_a\int C(\omega-f_a)
       p_o(a,\omega|\cQ_{\text{opt}}) d\omega\right\rangle_o \notag\\
  &= \left\langle\sum_a\int C(\omega-f_a)
       q(a|\omega, \cQ_{\text{opt}}) p_o(\omega) d\omega\right\rangle_o\notag\\
  &= \sum_a \int C(\omega-f_a) q(a|\omega, \cQ_{\text{opt}})
        \langle p_o(\omega)\rangle_o \,d\omega.
\label{eq:scpoF2}
\end{align}
Combining Eqs.~\eqref{eq:inverseTransform},\eqref{eq:scpoF1},\eqref{eq:scpoF2} then gives
\begin{align}
\langle S_C(p_o,F)\rangle_o 
  &\leq \sum_a \int C(\omega-f_a) q(a|\omega, \cQ_{\text{opt}})
        p(\omega) d\omega\notag\\
  &= C(p,\cQ_{\text{opt}}) = S_C(p,F),
\end{align}
proving the claim of the theorem.
\end{proof}

We note that Thm.~\ref{theorem:lowerBound} generalizes to arbitrary
oracle problems where the cost operators $A_a$ are
diagonal. Furthermore, the proof works for any family of probability
distributions $p_o$ such that $p$ is a mixture of the $p_o$.

Upper bounds on the optimal expected costs can be obtained by applying
the inequality in the next Theorem.

\begin{theorem}
\label{theorem:upperBound}
Let $\cQ$ be an algorithm for $S_C(r,F)$. Then
\begin{equation}
S_C(p,F)\leq \sum_a\int C(\omega-f_a) q(a|\omega,Q) p(\omega) d\omega.
\end{equation}
\end{theorem}

\begin{proof}
The right-hand-side is the expected cost for algorithm $\cQ$ given
oracle prior $p$.  Since $\cQ$ is optimal for prior $r$ but not
necessarily for $p$, this must be greater than $S_C(p,F)$,
which is the optimum expected cost for prior $p$.
\end{proof}

In view of the results of this section, we adopt the following procedure
$\cA(p,F)$ for estimating a lower bound and calculating an upper
bound of $S_C(p,F)$.
\begin{enumerate}
\item Choose discretization parameter $d$ and the number of random samples $k$.
  Large $d$ should tighten the bounds. Large $k$ improves the statistical
  estimate of the lower bound.
\item Independently choose $o_j\in(0,1/d)$, $j\in\{1,\ldots,k\}$
  uniformly at random.
\item Do the following for each $j\in\{1,\ldots,k\}$:
  \begin{enumerate}
  \item  Compute the optimum cost $C_j=S_C(p_{o_j},F)$, and 
    from the SDP solution, reconstruct an optimal algorithm $\cQ_j$ achieving
    this cost.
  \item  From $\cQ_j$ derive an algorithm for evaluating $q(a|\omega, \cQ_j)$.
  \item  Using this algorithm, evaluate
    \begin{equation*}
    \overline{C_j} = \sum_a\int C(\omega-f_a) q(a|\omega,\cQ) p(\omega) d\omega
    \end{equation*}
    by numerical integration.
  \end{enumerate}
\item Return $\frac{1}{k}\sum_j C_j$ as a statistical estimate of a
  lower bound (together with its estimated error) and
  $\min_j(\overline{C_j})$ as a numerical upper bound.
\end{enumerate}

\subsection{Procedure for Solving the Clock SDP}

Combining the ideas of this section, we use the following procedure
for approximately solving the clock SDP $S_C(p,\rls)$:
\begin{enumerate}
\item
  Choose a discretization $F$, $|F|=m$, of the frequency estimates. If the
  cost function is suitable, we can optimize the right-hand-side
  of Eq.~\ref{eq:qBound} and let $\epsilon_q$ be the corresponding
  querier discretization error bound.
\item
  Apply procedure $\cA(p,F)$ and let $c_l$ be the statistically estimated
  lower bound with estimated standard error $s_l$, and $c_u$ the
  numerical upper bound obtained.
\item 
  Give the estimated cost of $S_C(p,\rls)$ in the form
  $((c_l-\epsilon_q)\pm s_l, c_u)$.
\end{enumerate}
If the cost function is not suitable, we set $\epsilon_q=0$
and describe the discretization $F$ for which the bounds apply.  
		
\section{Results}
\label{sec:results}

We begin by considering one-query clock protocols ($t_f=1$).  In this
case the protocol consists of an initial query state to be prepared
and a final measurement.  Figures \ref{fig:buzekComparisonPeriodic}
and \ref{fig:buzekComparisonQuad} illustrate the importance of taking
into account prior knowledge when deriving optimal clock protocols.
Here our technique is compared to that of Ref.~\cite{buzek}, which
derives protocols under the assumption that $\omega$ is uniformly
distributed on $[-\pi,\pi]$. We consider Gaussian priors of various
widths and see that solving the clock SDP can substantially reduce the
expected cost.  If we use the periodic cost function $4
\sin^2(\frac{\omega - \omega_a}{2})$ considered in~\cite{buzek}, in
the limit of wide prior, we obtain identical protocols. This is
illustrated by Table~\ref{tab:initialStates}, which lists initial
$2$-atom states that optimize this cost function; the final row
corresponds to the state computed in Ref.~\cite{buzek}.

\begin{figure*}
\resizebox{1\textwidth}{!}{
\subfigure[\ One atom ]{
\includegraphics[width=.4\textwidth]{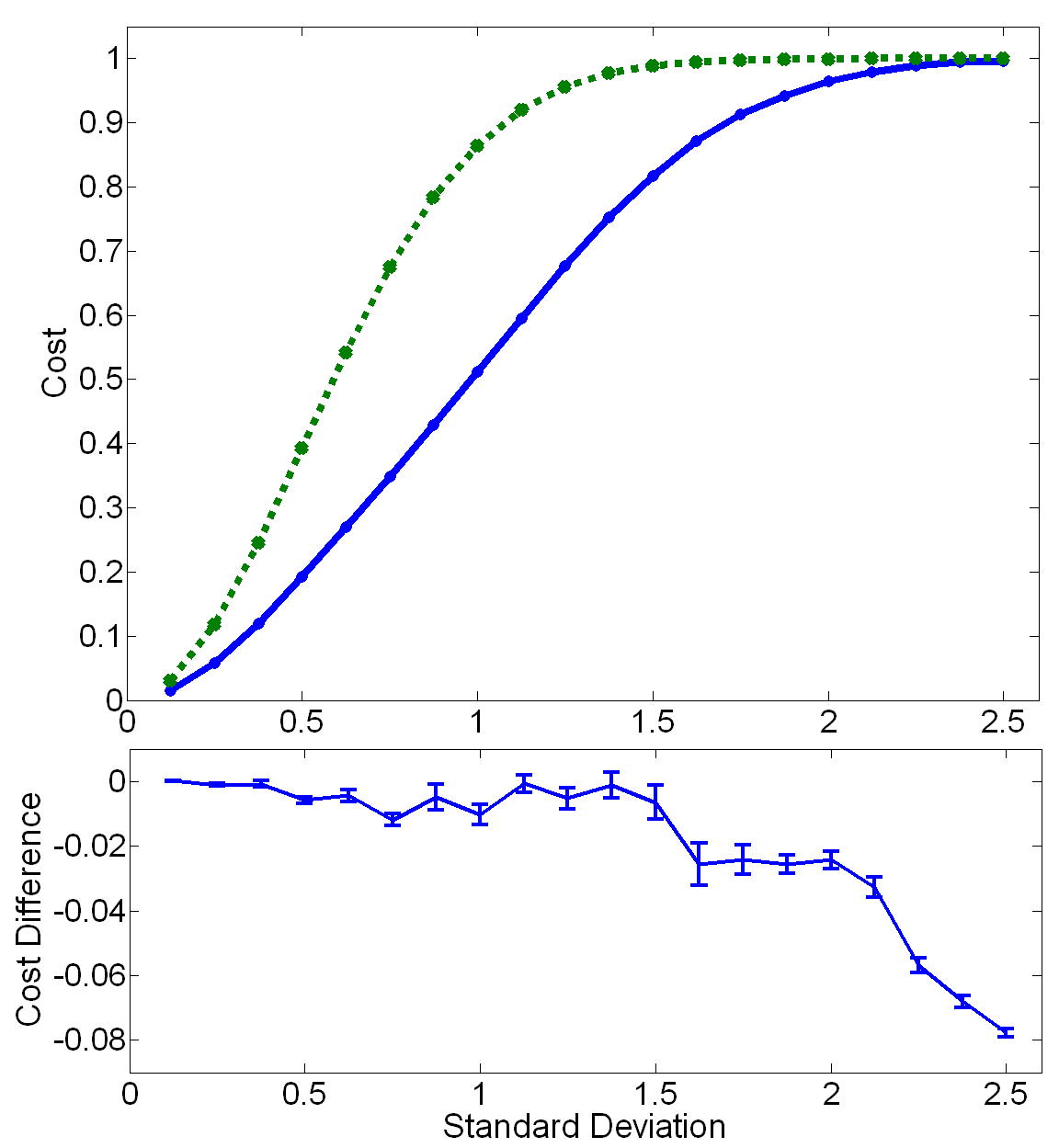}}
\subfigure[\ Two atoms ]{
\includegraphics[width=.4\textwidth]{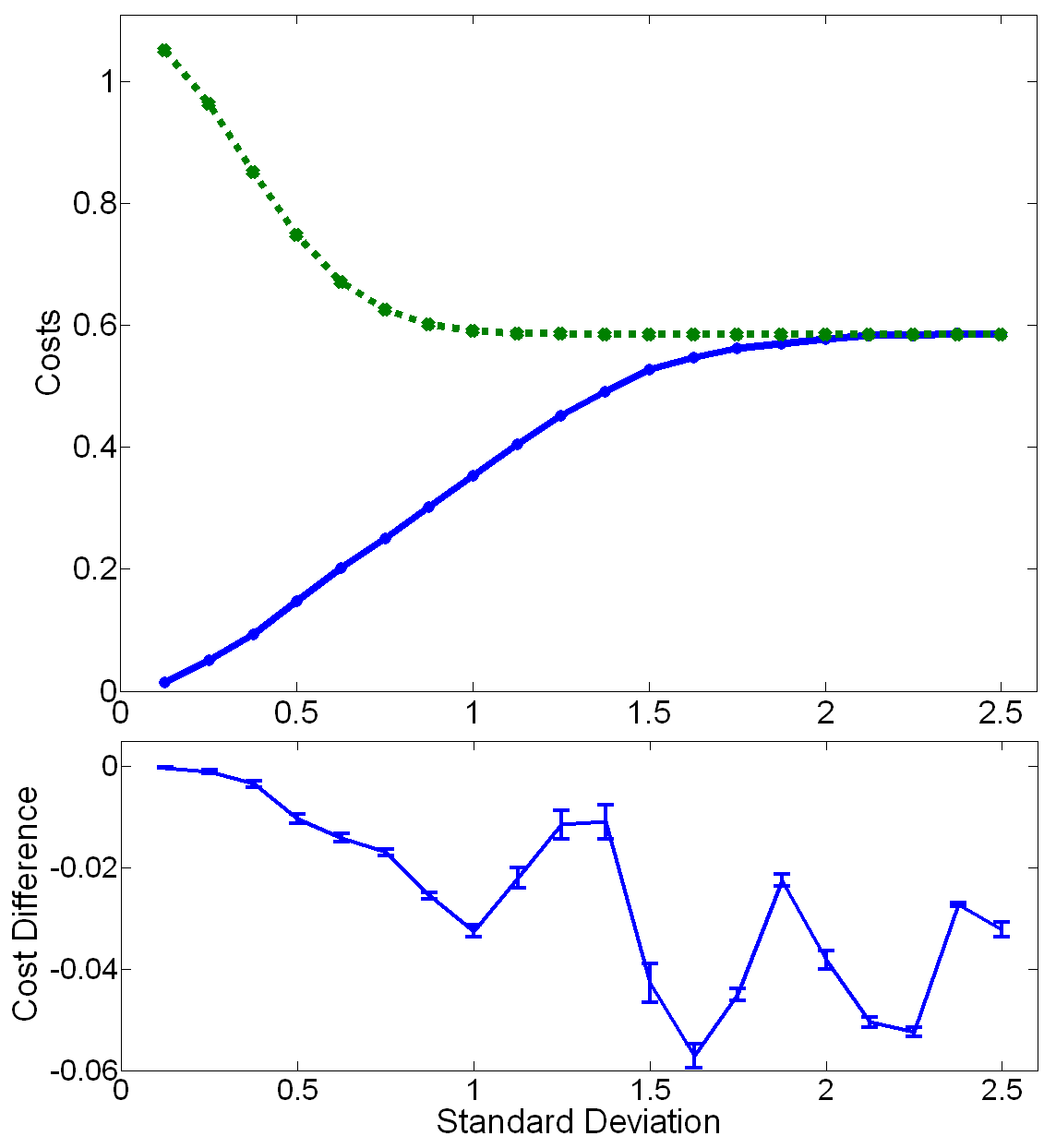}}
}

\caption{Comparison of the protocol derived in Ref. \cite{buzek} and
  those derived by our method for one and two atoms and the cost
  function $C(\omega - f_a) = 4\sin^2(\frac{\omega - f_a}{2})$.  We
  use a $15$-point oracle discretization, $20$ frequency estimates, and
  simulate Gaussian priors of various widths.  The figures on the top
  plot the cost computed in Ref. \cite{buzek} (dashed line) and our
  numerical upper bound (solid line), which as discussed, is
  equivalent to the minimum continuous cost obtained by one of our
  extracted algorithms.  The figures on the bottom plot the difference
  between our lower bound and upper bound, $c_l - c_u$, illustrating
  both the strength of our bounds and how much lower the continuous
  cost could potentially be.  The lower bound was computed by
  averaging $100$ discretizations; error bars show the estimated
  standard error of the average thus obtained.  Our querier
  discretization bounds cannot be applied to the periodic cost
  function used here, so the lower bounds are for the discretizations
  chosen.  The lines connecting the data points in our figures are to
  guide the eyes.}
\label{fig:buzekComparisonPeriodic}
\end{figure*}

\begin{figure*}
\resizebox{1\textwidth}{!}{
\subfigure[\ Two atoms ]{
\includegraphics[width=.4\textwidth]{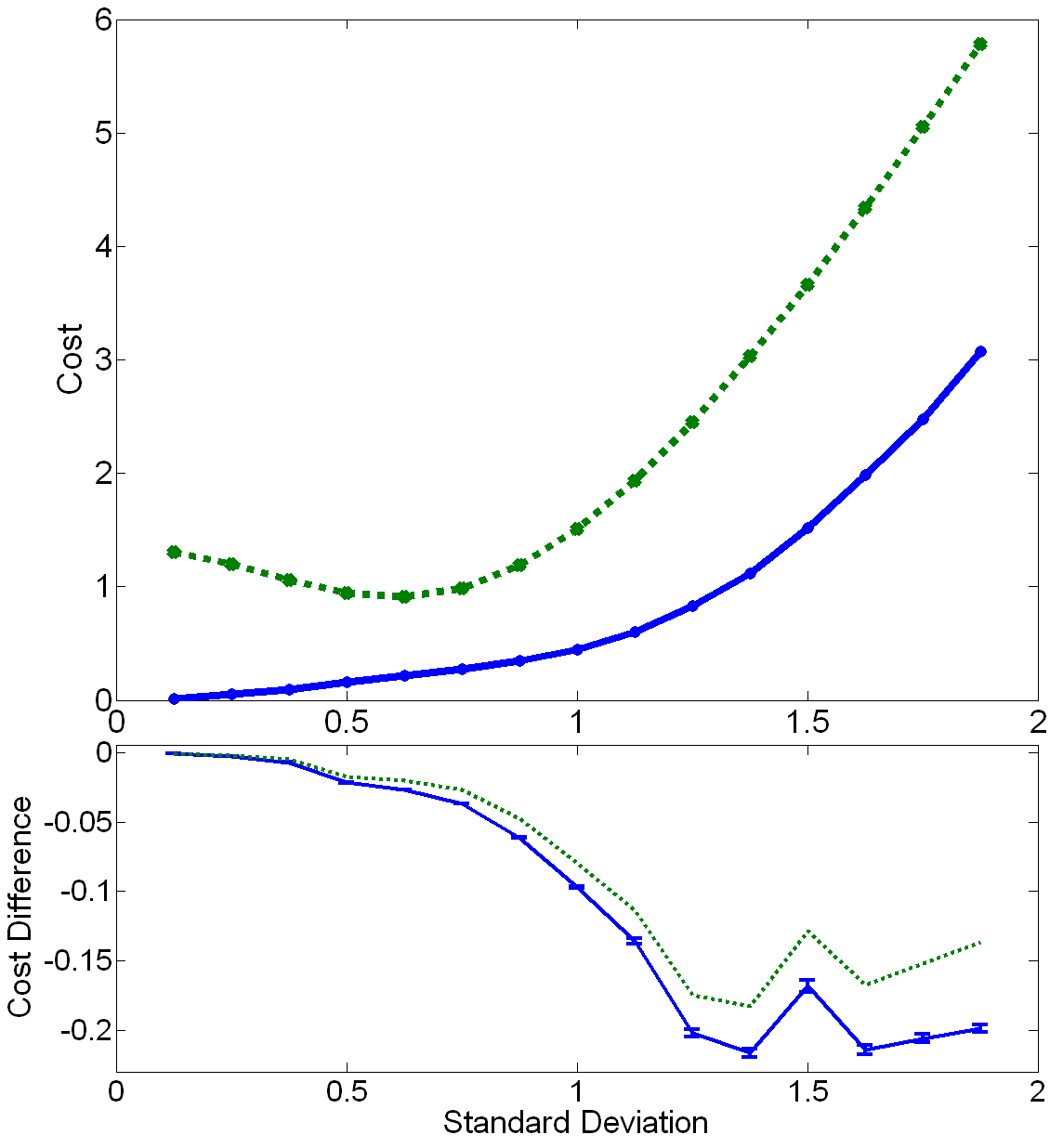}}
\subfigure[\ Four atoms]{
\includegraphics[width=.4\textwidth]{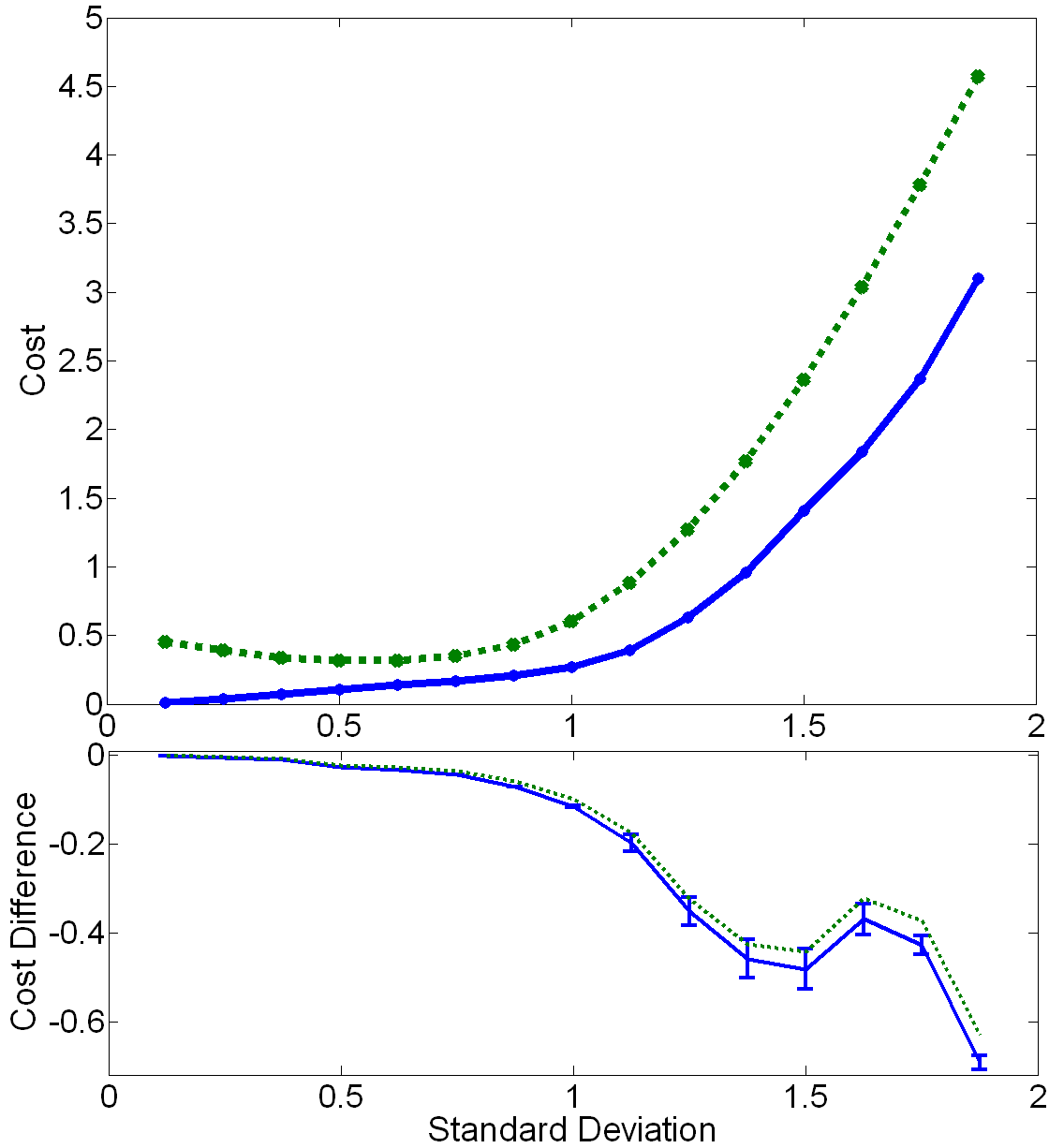}}
}

\caption{Comparison of the protocol derived in Ref. \cite{buzek} and
  those derived by our method for two and four atoms and the cost
  function $C(\omega - f_a) = (\omega - f_a)^{2}$.  SDP parameters are
  identical to those in Fig.~\ref{fig:buzekComparisonPeriodic}.  The
  plots are also as in Fig.~\ref{fig:buzekComparisonPeriodic}, except
  that for this cost function, we can compute querier discretization
  errors $\epsilon_q$.  Thus, the lower plots show $c_l-c_u$ as dotted
  lines above the line for $c_l-c_u-\epsilon_q$, which gives the lower
  bounds for the continuous problem. }
\label{fig:buzekComparisonQuad}
\end{figure*}

\begin{table}[!h]
\begin{tabular}{ c  | c  c  c}
Standard Deviation & $|0\rangle$ & $|1\rangle$ & $|2\rangle$ \\
\hline
.25 & .7071 & 0 & .7071 \\
.75 & .5626 & .6058 & .5626 \\
1.25 & .5170 & .6823 & .5170 \\
1.75 & .5025 & .7035 & .5025 \\
2.25 & .5000 & .7071 & .5000 \\
\end{tabular}
\caption{ Initial states that minimize the cost function $C( \omega -
  f_a) = 4 \sin^2(\frac{\omega - f_a}{2})$ for one-query protocols
  assuming different width Gaussian priors.  SDP parameters are
  identical to those of Fig.~\ref{fig:buzekComparisonPeriodic} }
	\label{tab:initialStates}
	\end{table}

For computing the graphs of Figs.~\ref{fig:buzekComparisonPeriodic}
and \ref{fig:buzekComparisonQuad}, and the initial states in
Table~\ref{tab:initialStates}, we did not optimize the frequency
estimates according to the iterative technique described at the end of
Sect.~\ref{sect:dErrorB}.  In order to achieve sufficiently small
discretization error, we discretized the prior frequency distribution
with $15$ points and used $20$ frequency estimates chosen by
minimizing Eq.~(\ref{eq:qBound}). The bounds are based on $100$ random
discretizations to obtain sufficiently good statistics on the lower
bound.

To verify the technique for obtaining error bounds, we compare our
upper and lower bounds to the optimal solution obtained according to
the formulas in Ref.~\cite{prior}; this is illustrated in
Fig.~\ref{fig:oracleComparison}.  We use the optimal set of classical
frequency estimates derived in Ref.~\cite{prior}; therefore, in the
continuous limit, the clock SDP and that of Ref.~\cite{prior} should
yield identical costs. Consequently, the deviation depicted in
Fig.~\ref{fig:oracleComparison} is due entirely to discretization
error and the limitations of our bounds.

\begin{figure}[ht]
\includegraphics[width=.5\textwidth]{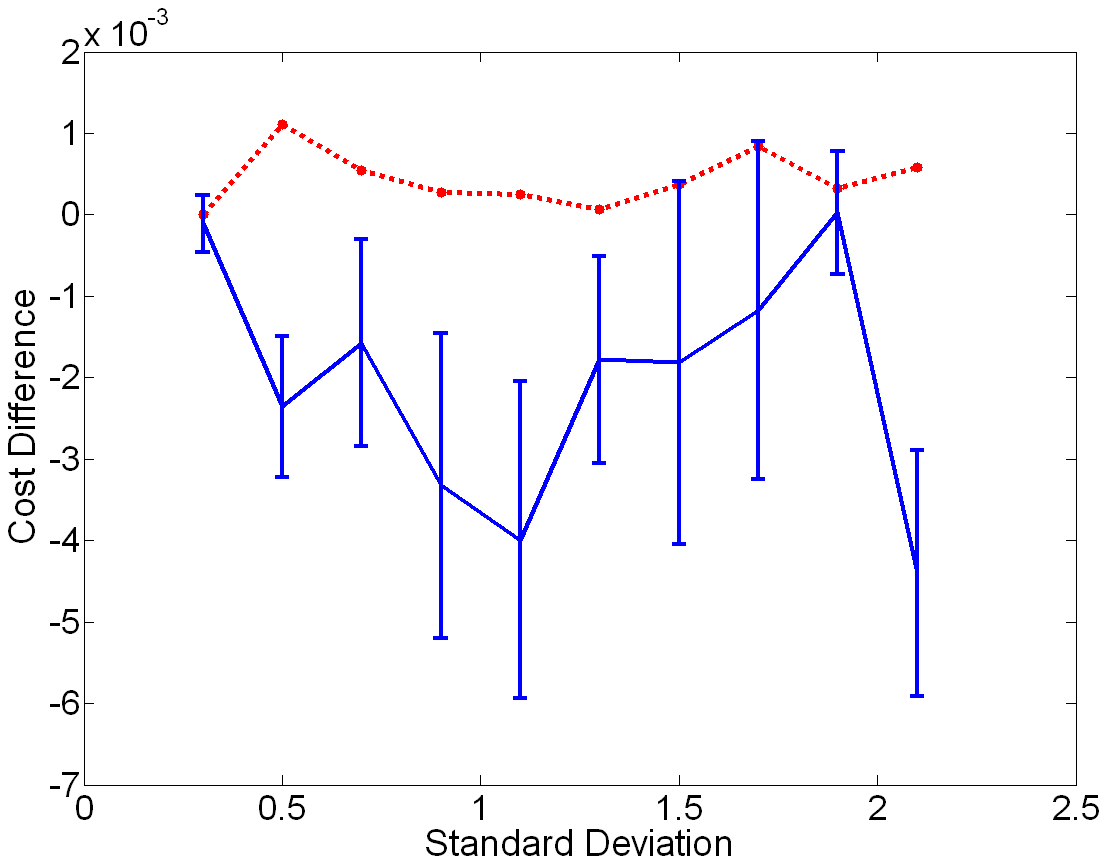}
\caption{Comparison of the optimal protocol derived in
  Ref.~\cite{prior} and the one obtained by our SDP for two atoms.  The
  dashed line depicts the difference between our upper bound and the
  cost computed in Ref.~\cite{prior}, and the solid line depicts the
  difference between our lower bound and the cost computed in
  Ref.~\cite{prior}.  Error bars again correspond to the estimated
  standard error in the lower bound. We used a $15$-point prior
  discretization and set $F$ to the optimal set of frequency estimates
  derived in~\cite{prior}, so there is no querier discretization
  error.  The lower bound was computed by averaging 500 random
  oracle discretizations.}
\label{fig:oracleComparison}
\end{figure}

As discussed, our procedure can also optimize sequences of two or more
clock queries ($t_f\geq 2$).  If these queries are fully coherent, the
algorithm is of the form given in Eq.~(\ref{eq:stateEvolution}), and
the SDP implicitly optimizes the initial state, the $U(t_i)$ and the
measurement.  Alternatively, we can combine the two queries
classically.  In this case we update our knowledge of the clock's
phase using Bayes' rule between the queries.  That is, after the first
query, we compute posterior distributions for each measurement
outcome, as in Eq.~(\ref{eq:posteriors}).  We then run the SDP
again, once for each outcome, using the corresponding posterior
distribution as the new prior.  We compute a new cost by averaging
each of the costs obtained, weighted by the probability of obtaining
the corresponding measurement outcome, $\tr(\sigma_a)$.  Here we
assume that there is no noise between sequential queries.  Any noise
would affect the intermediate prior distributions.
Fig.~\ref{fig:twoQueries} compares these two methods for a sequence of
two queries with two atoms.  Here we used $|F|=25$ possible frequency
estimates and the iterative technique described in
Sect.~\ref{sect:dErrorB} for optimizing them.  We found that the
number of frequency estimates needed was substantially reduced after
optimization.  For one query ($t_f=1$), and when such queries are
combined classically, three estimates per query suffice after applying
estimate optimization. When two queries are combined quantumly, five
estimates suffice. It appears that, as expected, while we gain
information by combining queries classically, fully coherent queries
provide the greatest advantage.  However, the technique developed for
computing lower bounds cannot be applied to classical combinations of
queries, so this advantage remains to be proven.
	
\begin{figure}[ht]
\includegraphics[width=.5\textwidth]{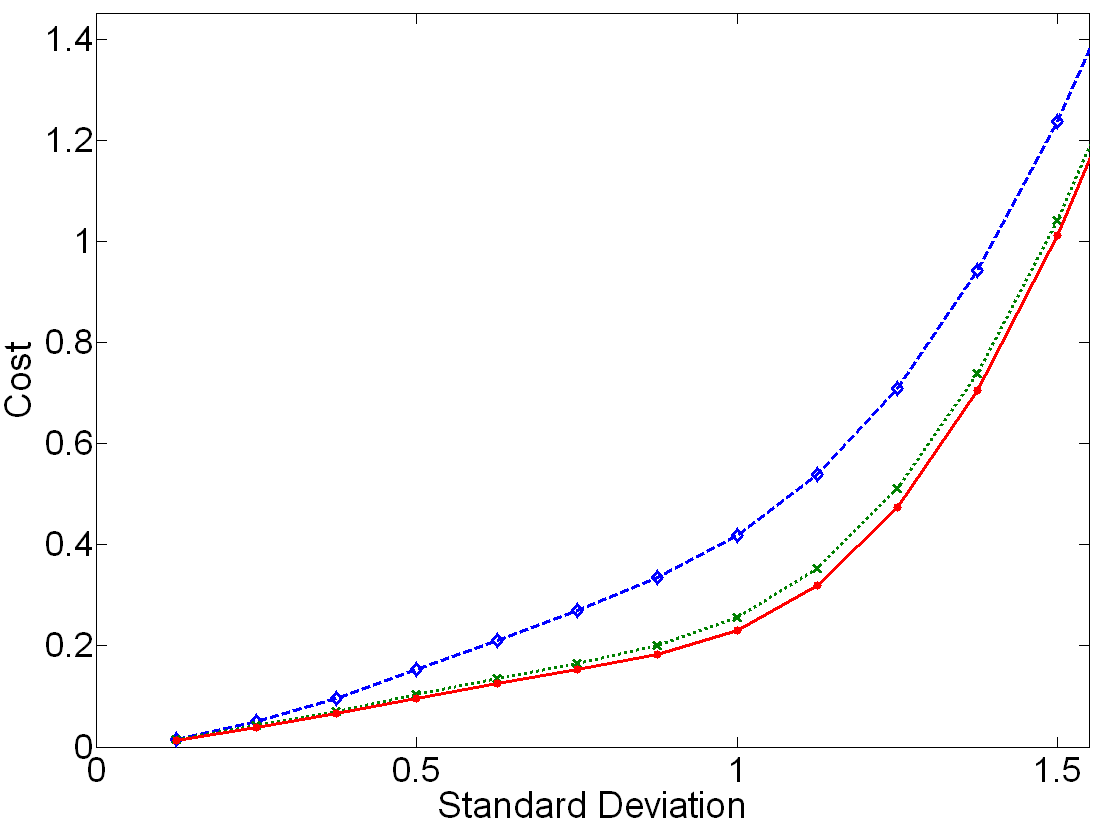}
\caption{Average cost after one query (top, dashed line), two queries
  combined classically (middle, dotted line), and two coherent quantum
  queries (bottom, solid line). We used a $15$-point oracle
  discretization, $|F| = 25$, the cost function $C(\omega - f_a) =
  (\omega - f_a)^2$, and simulated two atoms. }
\label{fig:twoQueries}
\end{figure}
	
Table \ref{tab:multipleQueries} examines the effect of adding more
atoms to the clock.  Additional atoms and additional queries always
provide an advantage, as can be seen by reading down the table.  The
far right column gives the computational time required to run the SDP
for a single discretization on a quad core $2.8\; \textrm{GHz}$
machine with $12\;\textrm{GiB}$ of RAM.  Note that we are able to
simulate more atoms than Table \ref{tab:multipleQueries} may imply.
For example, a simulation of $10$ atoms using a uniform $15$-point
oracle discretization and $t_f=1$ was computed in 6 minutes, 3
seconds, and yields a cost of $0.0785$. Computing bounds for a system
of this size, however, would require a great deal of computational
time.
	
\begin{table}[b]
  \begin{tabular}{ c  | c | c | c |c |c|c  }
	
		Number & Number & $c_l$ & $c_u$ & $s_l$ & $\epsilon_q$ & time  \\
                of Atoms & of Queries & & & & & (min:sec) \\
		\hline
		1 & 1 & .6010 & .6321 & .0127 & .0152 & 2:24\\
		2 & 1 & .4083 & .4379 & .0109 & .0164 & 2:33\\
		3 & 1 & .2885 & .3263 & .0105 & .0177 & 3:07\\
		4 & 1 & .1974 & .2563 & .0045 & .0192 & 2:46\\
		1 & 2 & .4144 & .4379 & .0132 & .0164 & 2:14\\
 		2 & 2 & .1957 & .2565 & .0047 & .0192 & 3:35\\
 		3 & 2 & .1071 & .2119 & .0020 & .0229 & 4:07\\
 		4 & 2 & .0902 & .2657 & .0022 & .0229 & 4:58
 		
  \end{tabular}
  \caption{Costs for various combinations of atoms and queries for a
    Gaussian prior with a standard deviation of $1$.  SDP parameters
    are identical to those in Fig.~\ref{fig:twoQueries}. The times are given
    for solving one instance of the SDP. $100$ instances were solved for
    the estimates of $c_l$ and $c_u$.}
  \label{tab:multipleQueries}
\end{table}

Notice that on the far right of Fig.~\ref{fig:buzekComparisonQuad}(b),
and in the last two rows of Table \ref{tab:multipleQueries}, the gap
between the lower bound and the upper bound becomes quite large.  We
need to increase the number of points in our discretizations if we
wish to compute better approximations to the optimal solution for the
continuous problem. This is as expected, since having more queries or
atoms enables finer resolution of oracle frequencies.
Fig.~\ref{fig:discretization} illustrates the effect on our bounds of
increasing the number of points in the oracle discretization.  Here we
are reanalyzing the last point in
Fig.~\ref{fig:buzekComparisonQuad}(b).

\begin{figure}[ht]
\includegraphics[width=.45\textwidth]{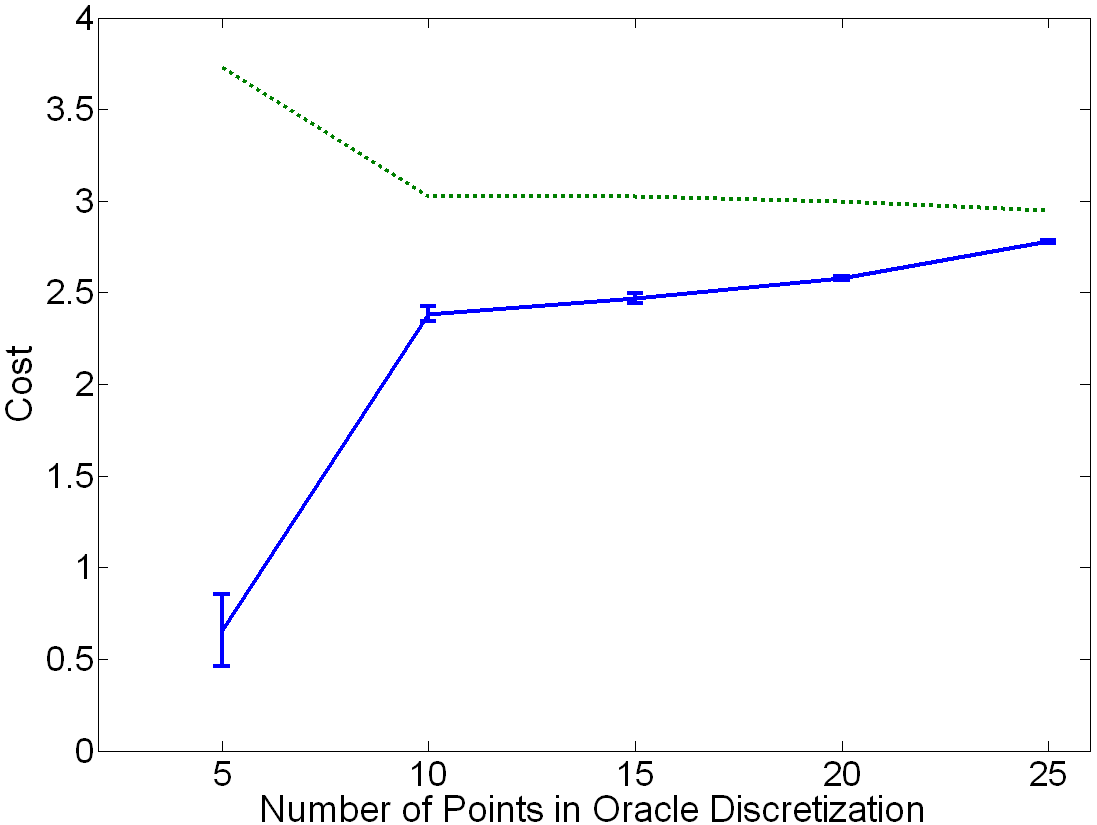}
\caption{Upper ($c_u$, dotted line) and lower ($c_l$, solid line)
  bounds on cost for oracle discretizations with varying numbers of
  points.  We are simulating the last point of
  Fig.~\ref{fig:buzekComparisonQuad}(b) - that is, 4 atoms, a standard
  deviation of 1.875, and the quadratic cost function.  The lower
  bound is computed by averaging 32 random oracle discretizations.
  Note that the values in Fig.~\ref{fig:buzekComparisonQuad} are for
  a $15$ point oracle discretization.}
\label{fig:discretization}
\end{figure}
	
In summary, the method presented here is a general way of deriving
better quantum clock protocols.  Discretization is necessary, but the
error introduced can be controlled. While the complexity of the SDPs
to be solved limits the application of the method to small quantum
systems, there are very promising atomic clocks, such as the ion-based
ones, that use only a small number of atoms.
	
The optimization strategy pursued here is greedy, taking into account
only the most recent prior. Therefore, it does not necessarily
minimize long-term variance. Further research is required to develop
more realistic strategies, taking advantage of knowledge of the clock's
history.

Finally, we note that the general form of the SDP and the
discretization analysis given here can be used to extend the adversary
method to continuous, noisy, and classical problems. It may also find
applications to quantum parameter estimation problems other than phase
or frequency estimation, where we wish to estimate the value of a
parameter $x$ of an arbitrary but known family of unitary operators
$U(x)$ that can be applied to quantum states of our choice.

\begin{acknowledgments}
This paper is a contribution of the National Institute of Standards
and Technology and not subject to U.S. copyright.  We thank Till
Rosenband and Didi Leibfried for pointing us to the clock problem.
\end{acknowledgments}

\bibliographystyle{apsrev}
\bibliography{clock}
	
\appendix

\section{Remote State Preparation}
\label{app:a}

Here we prove Thm.~\ref{the:RSP}. 

\begin{proof} 
Assume that we can remotely prepare $\{\sigma^O_a\}$ from $\rho^{OQ}$,
and let $\{P_a^Q\}_a$ be the required POVM. Then $\sum_a \sigma^O_a =
\sum_a\tr_QP_a^Q\rho^{OQ} = \tr_Q (\sum_a P_a)\rho^{OQ} = \rho^O$, by
the definition of a POVM.

The converse is a generalization of the GHJW theorem \cite{GHJW} to
mixed density operators.  Suppose that $\sum_a\sigma_a^O = \rho^O$.
To construct the required POVM, first write each
$\sigma^O_a$ as an explicit mixture of pure states
\[ 
\sigma_a^O = \sum_m p_{a,m} \ketbra{\psi_{am}}{\psi_{am}}.
\]
By assumption, $ \rho^O = \sum_{a,m} p_{a,m}
\ketbra{\psi_{a,m}}{\psi_{a,m}}$. By filling in mixture terms with
$p_{a,m}=0$ if necessary, we can assume that the range of the index
$m$ is independent of $a$. Since $\rho^{OQ}=\ketbra{\Phi}{\Phi}$ is pure, we can write the $OQ$
state in Schmidt form
\begin{equation}
\ket{\Phi} = \sum_j\sqrt{q_j}\ket{\phi_j}\ket{\varphi_j},
\end{equation}
where the $\ket{\phi_j}$ and $\ket{\varphi_j}$ are orthonormal in the
Hilbert spaces of $O$ and $Q$, respectively.  Thus $\rho^O$ can also
be written as the mixture $\rho^O=\sum_jq_j\ketbra{\phi_j}{\phi_j}$.
By unitary freedom (for example, see~\cite{nielsen:qc2001a}, pg. 103),
\begin{equation}
\sqrt{p_{am}}\ket{\psi_{am}} = \sum_{j} u_{am,j}\sqrt{q_j}\ket{\phi_{j}},
\end{equation}
where the $u_{am,j}$ are the entries of a unitary matrix.
We can now define 
\begin{equation}
P_a^Q = \sum_m\sum_{j,j'}u_{am,j}^* u_{am,j'}\ketbra{\varphi_j}{\varphi_{j'}}.
\end{equation}
That $\sum_a P_a = \one$ follows from the unitarity condition
for $u_{am,j}$. To verify the partial trace condition, compute
\begin{eqnarray}
\tr_Q P_a^Q\rho^{OQ} &=& 
  \sum_{m}\sum_{j,j'} u_{am,j}^* u_{am,j'}\notag\\
  &&\tr_Q\left(\ketbra{\varphi_j}{\varphi_{j'}}
    \sum_{l,l'} \sqrt{q_lq_{l'}} \ketbra{\phi_l}{\phi_{l'}}\ketbra{\varphi_{l}}{\varphi_{l'}}\right)
  \notag\\ 
  &=& \sum_{m}\sum_{j,j'}
u_{am,j}^* u_{am,j'}\sqrt{q_jq_{j'}}\ketbra{\phi_{j'}}{\phi_{j}}\notag\\
 &=&\sum_{m}p_{am}\ketbra{\psi_{am}}{\psi_{am}}\notag\\
 &=&\sigma_a^O,
\end{eqnarray}
as desired.
\end{proof}

\section{Querier Discretization Bound} 
\label{app:b}

Here is the proof of Thm.~\ref{thm:qBound}.

\begin{proof}
Define $C(p,\cQ)=\sum_a\int C(\omega-g_a)p(g_a,\omega|\cQ) d\omega$,
where the $g_a$ are $\cQ$'s frequency estimates.
Then $C(p,\cQ)$ is the expected cost of $\cQ$ given prior $p$.
Let $g(a)$ be defined by
\begin{equation} 
\label{eq:guesses} 
g(a) = \mathrm{argmin}_{g}\left( \int C(\omega - g)
  p(a|\omega,\cQ)p(\omega)d\omega \right).
\end{equation}
Then $g(a)$ is the optimum frequency estimate $\cQ$ could make given
measurement outcome $a$. Let $\cQ_g$ be $\cQ$ modified to make the
frequency estimates $g(a)$. 

Let $B$ be the expression on the right-hand-side of
Eq.~(\ref{eq:qBound}).  We show that $C(p,\cQ_g)\geq S_C(p,F)-B$ for any
algorithm $\cQ$. Since $S_C(p,\rls)=\inf_{\cQ} C(p,\cQ)=\inf_\cQ C(p, \cQ_g)$,
the result follows.  We prove the bound in two steps. In the first
step we force the frequency estimates to lie in $[f_1, f_N]$ and in the
second we change them to lie in $F$.

For the first step, let $\tilde g(a)$ be the value in $[f_1,f_N]$
nearest to $g(a)$.  If $\tilde g(a)=g(a)$, then $C(\omega-g(a))\geq
C(\omega-\tilde g(a))-M(\omega)$ since $M(\omega)\geq 0$. If $\tilde
g(a) = f_1$, then one of the following holds: 1. $\omega \geq f_1$, in
which case $\tilde g(a)$ is nearer $\omega$ and on the same side, so
that $C(\omega-g(a))\geq C(\omega-\tilde g(a)) \geq C(\omega-\tilde
g(a))-M(\omega)$.  2. $\omega < f_1$, in which case $C(\omega-g(a))
\geq 0 = C(\omega - f_1) - C(\omega - f_1) = C(\omega - \tilde g(a)) -
M(\omega)$.  A similar argument works for $\tilde g(a) =
f_N$. Substituting the inequalities in the integral for $C(p,\cQ)$
we get
\begin{eqnarray}
C(p,\cQ_g) &\geq& \sum_a\int C(\omega - \tilde g(a))p(a|\omega,\cQ)p(\omega)d\omega
   \notag\\&&{} - \sum_a\int M(\omega) p(a|\omega, \cQ)p(\omega)d\omega\notag\\
&=& C(p,\cQ_{\tilde g}) - \int M(\omega)d\omega .
\label{eq:qd1}
\end{eqnarray}

For the second step, we modify $Q_{\tilde g}$ to $Q_{\tilde f}$, where
$\tilde f(a)$ is one of the elements of $F$ on either side of $\tilde
g(a)$. That is, because $f_1\leq \tilde g(a)\leq f_N$, there exists a
unique $j$ such that $f_j\leq \tilde g(a)\leq f_{j+1}$, and we set
$\tilde f(a)$ to either $f_j$ or $f_{j+1}$. Define $\lambda\in[0,1]$
by $\tilde g(a) = \lambda f_j +(1 - \lambda) f_{j+1}$. It is
convenient to let $Q_{\tilde f}$ be a ``mixed'' (randomized)
algorithm, where $\tilde f(a)=f_j$ with probability $\lambda$ and
$f_{j+1}$ with probability $1-\lambda$. Note that a mixed algorithm of
this sort cannot be better than the optimal one, that is
$C(p,Q_{\tilde f})\geq S_C(p,F)$.  To bound the cost, we consider a
given $a$ and $\omega$ and estimate the quantity
\begin{eqnarray}
c(\omega,a)&=&
   \lambda C(w-f_j)+(1-\lambda) C(w-f_{j+1}) \notag\\
    &&{}- C(\omega - \tilde g(a))\notag\\
   &=& \lambda \left(C(w-f_j)-C(\omega-\tilde g(a))\right)\notag\\
       &&{}+ (1-\lambda) \left(C(w-f_{j+1})-C(\omega-\tilde g(a))\right).\notag\\
\label{eq:c_omega_a}
\end{eqnarray}
Define $\omega_0=\omega - \tilde g(a)$, $\omega_l=\omega-f_{j+1}$
and $\omega_u=\omega-f_{j}$. 
We can estimate
\begin{eqnarray}
C(\omega)-C(\omega_0) &=& (\omega-\omega_0)C'(\omega_0)\notag\\
            &&{}+ \int_{0}^{\omega-\omega_0}\int_{0}^{x} C''(\omega_0+y)dy dx
          \notag\\
          &\leq&
          (\omega-\omega_0)C'(\omega_0) \notag\\
           &&{}+ \frac{1}{2}(\omega-\omega_0)^2\max_y C''(y) \notag\\
          &\leq&
          (\omega-\omega_0)C'(\omega_0)
           + \frac{b}{2} (\omega-\omega_0)^2.\notag\\
\end{eqnarray}
Substituting this bound for each summand of Eq.~\eqref{eq:c_omega_a}
gives
\begin{eqnarray}
c(\omega,a) &\leq&
     \lambda( (\omega_u-\omega_0)C'(\omega_0)
                + \frac{b}{2}(\omega_u-\omega_0)^2 )\notag\\
     &&{}+ (1-\lambda)( (\omega_l-\omega_0)C'(\omega_0)
                + \frac{b}{2}(\omega_l-\omega_0)^2 )\notag\\
  &=&
     \frac{b}{2}\left(
          \lambda (\omega_u-\omega_0)^2 
          + (1-\lambda) (\omega_l-\omega_0)^2 )\right)\notag\\
  &=&
     \frac{b}{2} \left(
        \lambda (\omega_u-\omega_0)(\omega_u-\omega_l)
        \right)\notag\\
  &\leq&
     \frac{b}{2} \frac{(f_{j+1}-f_{j})^2}{4},
\end{eqnarray}
where we first applied $\lambda
\omega_u+(1-\lambda)\omega_l=\omega_0$.  The next identity requires
applying
$(1-\lambda)(\omega_l-\omega_0)=-\lambda(\omega_u-\omega_0)$ to the
second summand, and the
final inequality is obtained by noting that
$\lambda(\omega_u-\omega_0)/(f_{j+1}-f_{j})$ is maximized at
$\lambda=1/2$.  We can apply the above inequalities to bound
$C(p,\cQ_{\tilde g})$ as follows:
\begin{eqnarray}
C(p,\cQ_{\tilde g}) 
 &=& \sum_a\int C(\omega - \tilde g(a)) p(a|\omega, \cQ)
     p(\omega)d\omega \notag\\
 &\geq& \sum_a\int \big( \lambda C(\omega - f_j) + ( 1 - \lambda ) C(\omega - f_{j+1})\notag\\
 &&{} - \max_j \frac{b}{8}(f_{j+1}-f_j)^2
       \big) p(a|\omega, \cQ)p(\omega)d\omega \notag\\
 &=& C(p, \cQ_{\tilde f}) - \max_j \frac{b}{8}(f_{j+1}-f_j)^2 \notag \\
 &\geq& S_C(p,F) - \max_j \frac{b}{8}(f_{j+1}-f_j)^2.
\label{eq:qd2}
\end{eqnarray}
To finish the proof, we combine Eqs.~(\ref{eq:qd1}) and~(\ref{eq:qd2}).
\end{proof}

\end{document}